\documentclass[11pt,a4paper]{scrartcl}

\usepackage{amsmath,amssymb,amsthm}
\newtheorem{definition}{Definition}
\newtheorem{proposition}{Proposition}
\newtheorem{lemma}{Lemma}
\newtheorem{theorem}{Theorem}

\newtheorem{corollary}{Corollary}

\usepackage{fullpage}

\usepackage{enumerate}

\usepackage{xspace}

\newcommand{\cclass}[1]{\ensuremath{\mbox{\textup{#1}}}\xspace}
\newcommand{\BPP}{\cclass{BPP}}
\newcommand{\RP}{\cclass{RP}}
\newcommand{\coRP}{\cclass{coRP}}
\newcommand{\NP}{\cclass{NP}}
\newcommand{\coNP}{\cclass{coNP}}

\newcommand{\N}{\mathbb{N}}
\newcommand{\PP}{\mathcal{P}}
\newcommand{\OCT}{\textsc{Odd Cycle Transversal}\xspace}
\newcommand{\EBip}{\textsc{Edge Bipartization}\xspace}
\newcommand{\DFVS}{\textsc{Directed Feedback Vertex Set}\xspace}
\newcommand{\I}{\ensuremath{\mathcal I}\xspace}
\newcommand{\B}{\ensuremath{\mathcal B}\xspace}

\newcommand{\F}{\ensuremath{\mathbb F}\xspace}

\newcommand{\Q}{\ensuremath{\mathcal{Q}}\xspace}

\newcommand{\X}[1]{X'(#1)}

\newcommand{\containment}{\NP~$\subseteq$~\coNP/poly\xspace}

\newcommand{\BigOh}{\mathcal{O}}

\newcommand{\yes}{YES\xspace}
\newcommand{\no}{NO\xspace}

\title{Compression via Matroids: A Randomized Polynomial Kernel for Odd Cycle Transversal}
\author{Stefan Kratsch\thanks{Supported by the Netherlands Organization for Scientific Research (NWO), project ``KERNELS: Combinatorial Analysis of Data Reduction''.} \thanks{Utrecht University, Utrecht, the Netherlands,
\texttt{kratsch@cs.uu.nl}}
\and Magnus Wahlstr\"om\thanks{Max-Planck-Institute for Informatics, Saarbr\"ucken, Germany,
\texttt{wahl@mpi-inf.mpg.de}}}

\begin{document}

 \maketitle

\begin{abstract}
{\small 
The \textsc{Odd Cycle Transversal} problem (OCT) asks whether a given graph
can be made bipartite by deleting at most~$k$ of its vertices. In a
breakthrough result Reed, Smith, and Vetta (Operations Research Letters, 2004)
gave a~$\BigOh(4^kkmn)$ time algorithm for it, the first algorithm with
polynomial runtime of uniform degree for every fixed~$k$. It is known that this
implies a polynomial-time compression algorithm that turns OCT instances into
equivalent instances of size at most~$\BigOh(4^k)$, a so-called
kernelization. Since then the existence of a polynomial kernel for OCT, i.e.,
a kernelization with size bounded polynomially in~$k$, has turned into one of
the main open questions in the study of kernelization. Despite the impressive
progress in the area, including the recent development of lower bound
techniques (Bodlaender et al., ICALP 2008; Fortnow and Santhanam, STOC 2008) and
meta-results on kernelizations for graph problems on planar and other sparse
graph classes (Bodlaender et al., FOCS 2009; Fomin et al., SODA 2010), the existence of
a polynomial kernel for OCT has remained open, even when the input is
restricted to be planar.

This work provides the first (randomized) polynomial kernelization for OCT.
We introduce a novel kernelization approach based on matroid theory, where we
encode all relevant information about 
a problem instance into a matroid with a
representation of size polynomial in~$k$. 
For OCT, the matroid is built to allow us 
to simulate the computation of the iterative compression step of the algorithm of 
Reed, Smith, and Vetta, applied (for only one round) to an approximate odd cycle transversal which it is aiming to shrink to size~$k$.
The process is randomized with one-sided error exponentially small in~$k$,
where the result can contain false positives but no false negatives,
and the size guarantee is cubic in the size of the approximate solution.
Combined with an~$\BigOh(\sqrt{\log n})$-approximation (Agarwal et al., STOC 2005),
we get a reduction of the instance to size~$\BigOh(k^{4.5})$, implying a
randomized polynomial kernelization.  Interestingly, the known lower bound
techniques can be seen to exclude randomized kernels that produce no false
negatives, as in fact they exclude even co-nondeterministic kernels (Dell and van
Melkebeek, STOC 2010). Therefore, our result also implies that deterministic kernels
for OCT cannot be excluded by the known machinery. 
}
\end{abstract}

\section{Introduction}

One of the most successful (and natural) applications of parameterized
complexity is the study of combinatorially hard problems for the case that one
seeks a small solution. Such a problem is \emph{fixed-parameter tractable} (FPT)
if it can be checked in time~$f(k)n^c$ whether an instance of size~$n$ has a
solution of size at most (or at least)~$k$. When~$k$ is not too large, such an
algorithm can be considered efficient. This can be especially important for
minimization problems where the solution size corresponds to a real-world
cost. 

Curiously, for any decidable problem, having an FPT algorithm is known to
coincide with having a polynomial-time data reduction algorithm that reduces
any instance to an equivalent instance with size bounded by some function of~$k$, a so-called
\emph{kernelization}.\footnote{Indeed, running the 
  algorithm for~$n^{c+1}$ steps will either solve the instance, or allow the
  conclusion that~$f(k)n^c>n^{c+1}$ and consequently that~$n$ is bounded
  by~$f(k)$. The converse, a kernelization implying an FPT algorithm for any
  decidable problem, can be easily verified, although here the bound gets worse.}
However, 
the kernel size bound implied is the same
function~$f(k)$ as occurs in the running time bound, which for non-trivial
parameters will almost certainly be exponential in~$k$ unless~P~$=$~NP.

A more useful notion of efficient data reduction is \emph{polynomial}
kernels, i.e., kernelizations with kernel size bounded polynomially in the
parameter.  For many problems, this can be achieved by a direct study of
kernelization, e.g., the classic reduction of \textsc{Vertex Cover} to~$2k$
vertices~\cite{NemhauserT1975,ChenKJ01}, or the recent reduction of
\textsc{Feedback Vertex Set} to size~$\BigOh(k^2)$ by Thomass\'e~\cite{Thomasse10},
improving on work by Burrage et al.~\cite{BurrageEFLMR06} and Bodlaender~\cite{Bodlaender07}.  
Having small (polynomial) kernels provides a formalization of efficient data reduction,
and additionally, producing them often requires significant insight into the combinatorial
structure of a problem. 

Accordingly, the search for more and better kernelizations has evolved into a 
main branch of parameterized complexity (in fact, the opinion has been raised
that kernelization is what fixed-parameter tractability is really
about~\cite{Estivill-CastroFLR05}).  In particular, the existence of a
polynomial kernelization for a problem is seen as a significant threshold,
comparable to the existence of an FPT algorithm in the first place.  Recent
seminal work of 
Bodlaender et al.~\cite{BodlaenderDFH09} and 
Fortnow and Santhanam~\cite{FortnowS11} enforced the importance of this 
threshold by providing techniques to show that certain problems do not admit
polynomial kernels unless \containment and the polynomial hierarchy collapses
to its third level (see also Harnik and Naor~\cite{HarnikN10} for a related
question). Furthermore, a paper by Dell and van
Melkebeek~\cite{DellM10} was the first work to provide lower
bounds for the \emph{degree} of a polynomial kernelization; among other things,
their work implies an~$\BigOh(k^2)$ lower bound for
\textsc{Feedback Vertex Set} and \OCT.
Another recent focus has been \emph{meta kernelizations}, i.e.,
meta-level results that provide kernelizations for a large range of problems,
under restrictions on the input~\cite{BodlaenderFLPST09,FominLST10}
; see below. 

Still, for all this work, some problems have so far resisted classification
with respect to existence of polynomial kernels. Among these, emerging as the most important and most frequently raised
questions -- e.g., the two problems singled out as having the highest
importance at the recent workshop on kernelization, WorKer 2010 -- is the existence of 
polynomial kernelizations for the problems \OCT (OCT) and \DFVS (DFVS).
Both problems are also open even in the restricted case of planar
graphs~\cite{BodlaenderFLPST09-metakernelization-arxiv}. 
In this paper, we focus on OCT, where the question was first raised in~\cite{GuoGHNW06};
see also the recent survey on lower bounds for kernelization by Misra et al.~\cite{MisraRS11}. 

\paragraph{The Odd Cycle Transversal problem.}
The \OCT problem asks whether a given graph~$G$ can be made bipartite by
deleting at most~$k$ of its vertices.  Together with natural variants such as
\textsc{Edge Bipartization}, the edge deletion version, and \textsc{Balanced Subgraph}, 
the problem of removing odd-parity cycles in signed graphs, this problem has
numerous applications (see, e.g.,~\cite{Huffner09,HuffnerBN10}), and has received significant research
attention~\cite{ReedSV04,GuoGHNW06,FioriniHRV08,Guillemot08,Huffner09,LokshtanovSS09,HuffnerBN10,KawarabayashiR10-oct}.
With respect to parameterized and exact computation the breakthrough result
was the~$\BigOh(4^kkmn)$ time algorithm by Reed, Smith, and Vetta~\cite{ReedSV04}.\footnote{H\"uffner~\cite{NEWERVERSIONEXISTS_Huffner05} observed that the analysis could be improved to yield~$\BigOh(3^kkmn)$.}
This was the first
occurrence of the technique now called \emph{iterative compression}.
(Note that this is entirely distinct from the notion of \emph{instance compression}, 
which has been used as a term to generalize kernelization; see related work below.)
This technique also led to the first FPT algorithm for \DFVS by Chen et
al.~\cite{ChenLLOR08} and \textsc{Almost 2-SAT} by Razgon and O'Sullivan~\cite{RazgonO09}, and has become an important tool in parameterized
complexity; see the survey by Guo et al.~\cite{GuoMN09}. For \EBip, an~$\BigOh(2^km^2)$ time algorithm exists due to Guo et al.~\cite{GuoGHNW06}.
The best known approximation result for OCT is~$\BigOh(\sqrt{\log n})$ due to
Agarwal et al.~\cite{AgarwalCMM05-approx}, improving on earlier results with a ratio of~$\BigOh(\log n)$~\cite{GargVY96}.
For \EBip there is also
an~$\BigOh(\log \textrm{OPT})$-approximation by
Avidor and Langberg~\cite{AvidorL07-approx-klogk}.
Under the unique games conjecture, neither problem can have a constant-factor approximation~\cite{Khot02}.

\paragraph{What makes OCT special?} 
The OCT problem 
belongs to the class of graph modification problems,
i.e., finding a minimum number of modifications of vertices and/or
edges in a given graph to achieve a given property (like bipartiteness). 
If the target property is sparse (e.g., being a
forest for \textsc{Feedback Vertex Set}), or if it can be defined by a
constant number of forbidden structures (e.g., induced paths on three
vertices for \textsc{Cluster Editing}), then many such problems
are known to have a polynomial kernel, although the kernelization is
by no means always easy.  On the other hand, only few polynomial kernels are
known for target properties that do not have either characteristic.
One candidate is deleting at most~$k$ arcs to remove all directed cycles from a 
tournament\footnote{A directed graph is a tournament if for every pair
  of its vertices, exactly one of the two possible arcs is present.},
but here it suffices to consider the directed triangles~\cite{BessyFGPPST09}. 
Also, \textsc{Chordal Completion} has a polynomial kernel~\cite{NatanzonSS00}, but here,
a large obstacle is helpful in that a chordless cycle of length~$t$ requires
at least~$t-2$ edges to be added.
No such useful exception exists for OCT.

Additionally, there are involved meta-results for graph problems restricted to sparse inputs like planar, bounded genus, and~$H$-minor free graphs~\cite{BodlaenderFLPST09,FominLST10}. However, it can be easily seen that OCT is neither \emph{compact} nor \emph{quasi-compact} (\cite{BodlaenderFLPST09}), as \yes-instances do not have bounded treewidth (take an arbitrarily large grid graph and add a few edges that cause odd cycles); similarly it is not \emph{minor bidimensional} (\cite{FominLST10}) as the cost on a grid is zero. Thus, even for planar inputs, it is not covered by any known meta-result.
Generally, similar statements as the above can be made about DFVS.

\paragraph{Our work.}
In this paper, we give a randomized polynomial kernelization for OCT, for
unrestricted input.  
The kernel takes the form of a \emph{compression} of the instance into a
polynomial-sized matrix (with bounded entry lengths), such that the independent sets of columns in
the matrix reveal whether the instance is positive or negative.  By the \NP-completeness of OCT,
this then implies a randomized polynomial kernel. 
The result is produced by combining the iterative compression step of Reed, Smith, and Vetta~\cite{ReedSV04}
(in a suitable variant) with the theory of linearly represented matroids, specifically the class
called gammoids~\cite{Perfect1968}. 
We observe that, given a graph~$G$ and a set of terminals~$X$, a single gammoid can be used to produce a matroid that
encodes the flow from~$S$ to~$T$ in~$G-R$ for arbitrary~$S,T,R \subseteq X$,
and show (using results of Marx~\cite{Marx09-matroid}) how to produce a matrix representing this
matroid, of total size cubic in~$|X|$ and logarithmic in~$|G|$, in randomized polynomial time.
Having access to this information is sufficient to simulate the algorithm of~\cite{ReedSV04}. 
Here,~$X$ is an initial approximate solution to the problem. 
To get this initial set~$X$, we use the~$\BigOh(\sqrt{\log n})$-approximation of 
Agarwal et al.~\cite{AgarwalCMM05-approx} to produce
an~$\BigOh(\sqrt{\textrm{OPT}})$-approximation, giving an~$\BigOh(k^{4.5})$-size randomized
polynomial compression and completing the kernelization.
This approach of applying matroid theory to kernelization is a new tool in the field, and should
prove useful for several future kernelization results as well. 
This is also one of very few randomized polynomial kernelizations
(we are aware only of Harnik and Naor's probabilistic compression for
\textsc{Subset Sum}~\cite{HarnikN10}).
Our result also implies an~$\tilde \BigOh(k^3)$ compression for the
\textsc{Tanglegram Layout} problem from computational 
biology~\cite{FernauKP10} via the reduction given in~\cite{BockerHTW09};
a polynomial kernel for this problem was left open in~\cite{BockerHTW09}.

\paragraph{Related work.}
Generally, not much is known yet about excluding polynomial kernels
for graph modification problems, compared to the wide range of
problems that belong to this class.  Although a few kernelization
lower-bounds~\cite{DomLS09,KratschW09,DellM10,GuillemotPP10} and FPT
infeasibility results~\cite{Lokshtanov08} exist, for many problems in
this class, including \OCT and \DFVS, the question of polynomial kernels
is open. 
Another related set of problems is graph separation problems; here as well,
there are many FPT problems where the existence of polynomial kernels is
unknown, such as \textsc{Multiway Cut}~\cite{Marx06,ChenLL09} and
\textsc{Multicut}~\cite{BousquetDT11,MarxR11}.

Kernels for OCT for \emph{non-standard} parameters are studied by
Jansen and Kratsch~\cite{JansenK11}; a
polynomial kernel is obtained for the case that one is given an OCT instance~$(G,k)$ as well as a set~$X$ such that~$G-X$ is both bipartite and bounded
treewidth, with parameter~$|X|$.  The paper also contains related lower bounds (e.g.,
if~$G-X$ is treewidth~$2$ but not necessarily bipartite).

Harnik and Naor~\cite{HarnikN10} raised the question of compression of
\NP instances with respect to the witness size (e.g.,
size~$\BigOh(k\log n)$ for specifying a subset of size~$k$).
As they note, polynomial kernelization with the witness size as
parameter is equivalent to their notion of deterministic compression.
See Fortnow and Santhanam~\cite{FortnowS11} for further discussion
comparing the approaches (and note also that the factor of~$\log n$ can
be absorbed into~$k$ if a problem is FPT with running time~$2^{k^{\BigOh(1)}}n^c$).
Both Harnik and Naor~\cite{HarnikN10} and Fortnow and
Santhanam~\cite{FortnowS11} also give notions of probabilistic
compression; the one we use is closer to~\cite{FortnowS11}, though we
restrict ourselves to one-sided error (see Section~\ref{section:preliminaries}).
However, as there are almost no further examples of randomized
polynomial kernelization or compression, there will be plenty of time for settling 
notation later. 
In related work, compressions are instead called \emph{bikernels}~\cite{AlonGKSY10} or \emph{generalized kernelizations}~\cite{BodlaenderDFH09}.

Connections between parameterized complexity and matroid theory have
previously been studied by Marx~\cite{Marx09-matroid}, including a self-contained description of representation tools
and issues for matroids. For more on matroid theory and algorithmic aspects see Oxley~\cite{Oxley-matroidtheory} as well as Schrijver~\cite{Schrijver}.

\section{Preliminaries}\label{section:preliminaries}

\subsection{Parameterized complexity and kernelization}\label{section:prelim:kernels}

We use the following standard notation from parameterized complexity, for more
background on this area we refer the reader
to~\cite{DowneyF1998,FlumG06,Niedermeier06}. A \emph{parameterized problem}
over alphabet~$\Sigma$ is a language~$\Q\subseteq\Sigma^*\times\N$; the second
component of instances~$(x,k)\in\Sigma^*\times\N$ is called the parameter. A \emph{kernelization} (or \emph{kernel})
\emph{of~$\Q$} is a polynomial-time computable
mapping~$K:\Sigma^*\times\N\to\Sigma^*\times\N:(x,k)\mapsto(x',k')$ such
that~$(x,k)\in\Q$ if and only if~$(x',k')\in\Q$ and with~$|x'|,k'\leq h(k)$ where~$h$ is a computable function;~$h$ is called the \emph{size} of the kernel.
A kernelization is a \emph{polynomial kernelization} if the size~$h(k)$ is 
polynomially bounded.

We use the term \emph{(parameterized) compression of~$\Q$ (into~$\Q'$)}
to denote the relaxed variant where~$K$ is allowed to map to a different
language~$\Q'$ (also called \emph{bikernel}~\cite{AlonGKSY10} or \emph{generalized kernelization}~\cite{BodlaenderDFH09}).
When~$\Q'$ is in \NP and~$\Q$ with parameter coded in unary is \NP-complete then, using the implicit Karp reduction, a polynomial compression of~$\Q$ into~$\Q'$ implies a polynomial kernelization for~$\Q$~\cite{BodlaenderTY09}.

We define a natural randomized version of kernelization with one-sided error,
corresponding to the complexity class \coRP
(variants for \RP and \BPP could be defined similarly).
Our notion of polynomial \coRP-compression is essentially equivalent to that of probabilistic compression in~\cite{FortnowS11},
except for our one-sided error, and that~\cite{FortnowS11} defines the parameter to be given in unary.

\begin{definition}[\coRP-kernelization]
Let~$\Q\subseteq\Sigma^*\times\N$. 
A randomized polynomial-time algorithm~$K$ with inputs and outputs in~$\Sigma^*\times\N$ is a \emph{randomized kernelization without false negatives}, or \emph{\coRP-kernelization}, for~$\Q$ if there is a computable function~$h:\N\to\N$ such that for all~$(x,k)\in\Sigma^*\times\N$:
\begin{enumerate}
\item if~$(x,k)\in\Q$ then~$prob[K((x,k))\in\Q]=1$,
\item if~$(x,k)\notin\Q$ then~$prob[K((x,k))\notin\Q]\geq\frac{1}{2}$, and
\item the size of~$x'$ and the value of~$k'$ are bounded by~$h(k)$, where~$(x',k'):=K((x,k))$.
\end{enumerate}
The notions of \coRP-compression and polynomial \coRP-compression are defined in the natural way.
\end{definition}

Note that unlike algorithms for \BPP, \RP, or \coRP we cannot use majority,
disjunction, or conjunction over the outputs of~$N$ independent runs to boost the success probability, since kernelizations and compressions typically do not solve instances.
Harnik and Naor~\cite{HarnikN10} observed that a similar effect may be attained by making a combined instance from the result of~$N$ independent runs of the compression (e.g., in our setting, creating an output which is to be interpreted as~$((x_1,k_1)\in \Q) \land \ldots \land ((x_t,k_t) \in \Q)$).
Strictly speaking, this approach gives a compression,
but it can again be turned into a kernelization by the argument via the Karp reduction.

\subsection{Matroids}

Matroids are interesting combinatorial structures, generalizing the notion of independence from 
linear algebra, while also drawing from graph theory. 
There is an extensive theory of matroids, as well as several important algorithmic results; 
see Oxley~\cite{Oxley-matroidtheory} and Schrijver~\cite{Schrijver}.

A \emph{matroid} is a pair~$M=(E,\I)$, where~$E$ is the \emph{ground set} and~$\I \subseteq 2^E$ 
a collection of \emph{independent sets}, 
such that: (i) $\emptyset \in \I$; (ii) if~$I_1 \subseteq I_2$ and~$I_2 \in
\I$, then~$I_1 \in \I$; and (iii) if~$I_1,I_2 \in \I$ and~$|I_2|>|I_1|$, then
there exists some~$x \in (I_2 \setminus I_1)$ such that~$I_1 \cup \{x\} \in \I$.
A set~$I \subseteq E$ is \emph{independent} if~$I \in \I$, and \emph{dependent} otherwise.
A set~$B \in \I$ is a \emph{basis} of~$M$ if no superset of~$B$ is independent;
a matroid may equivalently be defined by its set of bases (among other variants).
Let~\B be the set of bases of~$M$, and~$\B^*=\{E \setminus B: B \in \B\}$.
Then~$\B^*$ is the set of bases of a matroid~$M^*$, called the \emph{dual} of~$M$.
Note that~$(M^*)^*=M$.

Let~$A$ be a matrix over a field~$\F$ and~$E$ be the set of columns of~$A$.
Let~$\I$ be the set of all sets~$X \subseteq E$ of columns that are linearly independent over \F (as vectors).
Then~$(E,\I)$ defines a matroid~$M$, and we say that~$A$ \emph{represents}~$M$.
A matroid is \emph{representable} (over a field \F) if there is a matrix (over \F) that represents it.
A matroid representable over some field is called \emph{linear}.
In this work, we will concern ourselves only with linear matroids.
From a representation of~$M$, one can easily get a representation of~$M^*$
over the same field.

Finally, we define \emph{minors} of a matroid.  For a matroid~$M=(E,\I)$ and a set~$T \subseteq E$,
\emph{deleting}~$T$ results in a matroid~$M\setminus T = (E \setminus T,\I')$ where~$\I'=\{I \in \I: I \subseteq E\setminus T\}$.
\emph{Contracting}~$T$ results in a matroid~$M/T=(M^*\setminus T)^*$; if~$T \in \I$, then the independent sets
of~$M/T$ are the sets~$X \subseteq E\setminus T$ such that~$X \cup T \in \I$. 
A \emph{minor} of a matroid~$M$ is any matroid produced from~$M$ by deletions and contractions.
Both operations can be performed with preserved representation.

\section{Polynomial encoding of terminal cuts using gammoids}\label{section:newmatroidsection}

The basic situation that is handled in this section is the following.
Let~$D=(V,A)$ be a directed graph, and let~$X \subseteq V$ be a set of terminals.
We want to reduce the graph to a size polynomial in~$|X|$ and~$\log |V|$, while
preserving the size of a minimum vertex cut~$(S,T)$ for all sets~$S, T \subseteq X$.
Here, a vertex cut is understood as being allowed to delete vertices of~$S$ or~$T$ as well as
other vertices of~$V$; thus the min-cut sizes are bounded by~$|X|$. As an extension, we will also 
require that we may specify any set~$R\subseteq X$ as removed, i.e., we want to have also the 
cuts~$(S,T)$ in~$D-R$.

Clearly, this question is closely connected to the search for polynomial kernels for FPT cut problems.
However, a direct combinatorial reduction to achieve this, e.g., via edge contractions and vertex deletions or
other direct simplifications on the graph, seems difficult.  It is not even clear whether
there always exists a graph of the required size, where every minimum~$(S,T)$-cut for~$S,T
\subseteq X$ has the same cardinality as in~$D-R$.
Instead, we here solve the question by introducing the use of matroids and matroid representations
to the field of kernelization.  

Let us recall a few helpful definitions.  For~$S,T\subseteq V$, the set~$T$ is 
\emph{linked to}~$S$ if there exist~$|T|$ vertex-disjoint paths from~$S$ to~$T$, 
where also the end points of the paths must be disjoint. The sets~$S$ and~$T$ do not need to be disjoint; a
vertex is linked to itself by a path of length zero.
By the cut-flow duality, it is clear that being able to find the linked subsets 
of~$X$ will suffice to answer all questions about cuts~$(S,T)$ in~$D$. 
Perfect~\cite{Perfect1968} showed, given any~$D=(V,A)$ and~$S\subseteq V$, that the subsets of~$T$ which are linked to~$S$ in~$D$ form a matroid~$(D,S)$, 
of a class now called \emph{gammoids} (see~\cite{Oxley-matroidtheory,Schrijver}).
Marx~\cite{Marx09-matroid} gave a randomized polynomial-time procedure for finding a 
representation of this matroid.  

\begin{theorem} [\cite{Perfect1968,Marx09-matroid}]\label{theorem:perfectmarx}
Let~$D=(V,A)$ be a directed graph, and let~$S\subseteq V$. The subsets~$T\subseteq V$ 
which are linked to~$S$ form the independent sets of a matroid over~$V$. Furthermore, a 
representation of this matroid can be obtained in randomized polynomial time with one-sided error.
\end{theorem}

Here, one-sided error means that dependent (non-linked) sets are preserved, but independent
(linked) sets in the graph may not be, i.e., if the procedure returns a matrix~$A$, then there may
be some subsets of~$T$ which are linked to~$S$ but which are not independent in the matroid
represented by~$A$.  However, the risk of this can be made arbitrarily small.  

It remains for us only to bound the bit-length of the entries of the matrix (which would otherwise
be polynomial in~$|V|$).  This is easily done by standard methods. 

\begin{corollary} \label{cor:smallgammoid}
Let~$D=(V,A)$ be a directed graph,~$\epsilon>0$ a given real, and let~$S$ and~$T$ be possibly
overlapping subsets of~$V$. Let~$M$ be the gammoid formed by subsets of~$T$ linked to~$S$.
A representation of~$M$ as an~$|S|\times|T|$ matrix over the rationals with entries of bit-length 
$\BigOh(\min(|T|,|S|\log |T|) + \log (1/\epsilon) + \log |V|)$
can be computed in randomized polynomial time with one-sided error at most~$\epsilon$.
\end{corollary}
\begin{proof}
Theorem~\ref{theorem:perfectmarx} can be made to return an~$|S|\times |V|$ matrix over the rationals, with
arbitrarily small one-sided error~$\epsilon'>0$ and individual entries being integers of bit-length polynomial
in~$|V|$ and~$\log 1/\epsilon'$.  Let~$\epsilon'=\epsilon/2$, and let~$A$ be the matrix returned,
with columns not in~$T$ removed.
To reduce the length of the entries, we take all entries in~$A$ modulo a sufficiently large random prime~$p$.
We argue that the matrix~$A'$ produced this way satisfies all conditions.

Consider an independent column set in~$A$.  Since independence corresponds to a square submatrix
with non-zero determinant, we see that~$A$ and~$A'$ differ in this aspect only if~$p$ divides said determinant.
The number of distinct prime factors in a number is bounded by the bit-length, which for our
determinants is polynomially bounded in~$|V|+\log 1/\epsilon$.  Since the number of maximal
independent sets in~$M$ is bounded by both~$|T|^{|S|}$ and by~$2^{|T|}$, the total number of
distinct primes is bounded as~$t=\min(|T|^{|S|},2^{|T|})(|V|+\log 1/\epsilon)^{\BigOh(1)}$.
Thus, if we pick a random prime from a set of
size at least~$t'=(2/\epsilon)\cdot t$, the total risk of failure is bounded by~$\epsilon$.
By the Prime Number Theorem, primes of bit length~$\BigOh(\log (t' \log t'))=\BigOh(\log t+\log 1/\epsilon)$ are sufficient for this, 
which matches the statement of the corollary. 
By the AKS primality testing algorithm~\cite{AgrawalKS02-primes}, finding a random prime can be
done with high probability by repeated uniform sampling, and if this fails we may pick an arbitrary fixed prime.
It can be seen that errors throughout are one-sided.
\end{proof}

The following proposition extends the available gammoid structure to allow any subset~$S$ of the terminals as sources (i.e., without fixing it in advance), and to support also deletion of terminals. This will be our interface for using Corollary~\ref{cor:smallgammoid} in the \OCT kernelization. The argument is straightforward and works as well for a given directed graph. 

\begin{proposition}\label{proposition:interface}
Let~$G=(V,E)$ be an undirected graph, let~$X\subseteq V$ be a set of terminals, and let~$X':=\{v'\mid v\in X\}$ be a set of new vertices. There is a polynomial-time construction of a directed graph~$D=(V\cup X',A)$ such that~$I\subseteq X\cup X'$ is an independent set of the gammoid~$(D,X')$ if and only if~$T$ is linked to~$S$
in~$G-R$ where
\begin{itemize}
\item $S$ contains all vertices~$v\in X$ with~$v,v'\notin I$,
\item $T$ contains all vertices~$v\in X$ with~$v,v'\in I$, and
\item $R$ contains all vertices~$v\in X$ with~$v\in I$ but~$v'\notin I$.
\end{itemize}
\end{proposition}

\begin{proof}
The arc set~$A$ of the digraph~$D$ is defined as follows: For any two adjacent vertices~$u,v\in V$ add~$(u,v)$ and~$(v,u)$ to~$A$. Then add an arc~$(v',v)$ for all~$v\in X$ (and corresponding~$v'\in X'$).

We consider first any independent set~$I\subseteq X\cup X'$ 
of the gammoid~$(D,X')$. There are~$|I|$ vertex-disjoint directed paths from~$X'$ to~$I$ in~$D$; fix any such packing~$\PP$ of directed paths. By the structure of~$D$ all paths of~$\PP$ are either 
of form~$(u')$ with~$u'\in X'$, or of form~$(u',u,\ldots,v)$ (possibly with~$u=v$) and containing no vertices of~$X'\setminus\{u'\}$.
Let~$P=(u',u,\ldots,v)\in\PP$ with~$u'\in X'$ and~$v\in T$, i.e., with~$v,v'\in I$. Clearly,~$u'\neq v'$ since~$(v')\in\PP$ is the unique directed path from~$X'$ to~$v'$ and must be contained in~$\PP$ as~$v'\in I$ (and using vertex-disjointness). Since~$u$ and~$u'$ are on~$P$, no other path of~$\PP$ can end in~$u$ or~$u'$, thus~$u,u'\notin I$ and hence~$u\in S$. Finally, no vertex~$p\in R$ can be on~$\PP$ since, by~$p\in I$, that requires another path of~$\PP$ to end in~$p$.
Now, the subpath~$(u,\ldots,v)$ contained in~$D-X'$
corresponds to an undirected path from~$S$ to~$T$ in~$G-R$, and all those paths are vertex-disjoint.

Now, let~$\PP$ be a set of~$|T|$ 
vertex-disjoint paths from~$S$ to~$T$ in~$G-R$. We construct a set~$\PP'$ of directed vertex-disjoints paths from~$X'$ to~$I$ in the digraph~$D$ where~$I$ is obtained according to the statement of the proposition. For each path~$(u,\ldots,v)\in\PP$ we add the path~$(u',u,\ldots,v)$ to~$\PP'$. Clearly, those paths exist in~$D$ and they are vertex-disjoint. We also require paths ending in the vertices of~$I\setminus T$. This includes vertices~$r\in R$ and vertices~$v'$ with~$v\notin S\cup R$. It is easy to see that adding paths~$(r',r)$ and~$(v')$, respectively, for all those vertices yields the required path packing~$\PP'$ (key fact: in the initial~$|T|$ paths we only used~$v'$ when~$v\in S$, and vertices~$r\in R$ were unused). Thus~$I$ is an independent set of the gammoid~$(D,X')$. 
\end{proof}

Since it appears a very useful form, e.g., for obtaining polynomial kernels for other cut problems, we explicitly state the combination of Proposition~\ref{proposition:interface} and Corollary~\ref{cor:smallgammoid} as a corollary.

\begin{corollary}\label{corollary:smallrepresentedmatroid}
Let~$G=(V,E)$ be an undirected graph,~$X\subseteq V$ a set of
terminals, and~$\epsilon$ a positive real. There is a randomized
polynomial-time algorithm computing an~$|X|\times 2|X|$ matrix with integer entries of bit length~$\BigOh(
|X|+\log |V| + \log 1/\epsilon)$, such that with probability at least~$(1-\epsilon)$
any set of columns~$I\subseteq X\cup X'$ is independent in~$M$ if~$T$ is linked to~$S$ in~$G-R$, where~$S,T,R\subseteq X$ are defined as in Proposition~\ref{proposition:interface}. The error is one-sided: the number of disjoint paths as indicated by independence is a lower bound on the true value in~$G$.
\end{corollary}

\section{A randomized polynomial kernel for odd cycle transversal}\label{section:oct}

In this section, using the results presented in the previous section, 
we will give our randomized polynomial kernelization for \OCT. 
We will start by a presentation of the FPT algorithm of Reed, Smith, and 
Vetta~\cite{ReedSV04}, as understanding this algorithm is critical to
understanding the kernelization.  This is presented in
Section~\ref{section:rsvalg}.  Then we present the kernelization in
Section~\ref{section:actualkernel}, and finally discuss the relation with
lower bounds in Section~\ref{section:kernellower}.

\subsection{The Reed-Smith-Vetta algorithm} \label{section:rsvalg}

The FPT algorithm of Reed, Smith, and Vetta~\cite{ReedSV04} solves \OCT by a recursive approach: Solve the problem for~$G-v$, where~$v$ is an arbitrary vertex. If it returns a solution~$X_v$ of size at most~$k$, then~$X:=X_v\cup\{v\}$ is a solution of size at most~$k+1$ for~$G$, and the following \emph{compression} version of the problem is solved. Otherwise~$(G-v,k)$ is \no, thus~$(G,k)$ must be \no.
\begin{quote}
{\bf Input:} A graph~$G=(V,E)$, an integer~$k$, and a bipartization set~$X$ of size~$k+1$. 
\\{\bf Parameter:} $k$.  
\\{\bf Question:} Is there a bipartization set~$Y$ for~$G$ such that~$|Y|\leq k$?
\end{quote}
The compression routine in the Reed-Smith-Vetta algorithm consists of
trying exhaustively all ways of how the set~$X$ could interact with a smaller solution~$Y$, each coming down to a maximum flow
computation. Concretely, we create a graph~$G'=(V',E')$ from~$G$ and~$X$ in the following way:
let~$S_1 \cup S_2$ be a bipartition of~$G-X$.  Let~$V'=V-X \cup \{x_1,x_2: x \in X\}$, where~$x_1$ and~$x_2$ are new vertices.
Connect~$x_1$ to all neighbors of~$x$ in~$S_2$ and~$x_2$ to all neighbors of~$x$ in~$S_1$. By subdividing edges,
we may assume that there are no edges inside~$X$. Note that~$G'$ is bipartite with partitions~$S_1 \cup \{x_1: x\in X\}$ 
and~$S_2 \cup \{x_2: x \in X\}$. For~$U \subseteq X$, let~$\X{U}:=\{x_1,x_2: x \in U\}$, and let~$X'=\X{X}$.

The algorithm searches for cuts through~$X'$ in~$G'$.  For a subset~$U \subseteq X$, let a pair~$(S,T)$ of disjoint subsets of~$X'$ be a \emph{valid split} 
of~$U$ if for every~$x \in U$ we have~$|\{x_1,x_2\} \cap S| = |\{x_1,x_2\} \cap T|=1$ and 
for every~$x \in (X \setminus U)$ we have~$|\{x_1,x_2\} \cap S| = |\{x_1,x_2\} \cap T|=0$.
The following lemma is a direct consequence of~\cite{ReedSV04}.

\begin{lemma}[\cite{ReedSV04}] \label{lemma:rsv}
Let~$G=(V,E)$ be a graph and let~$X\subseteq V$ such that~$G-X$ is bipartite. Let~$G'$ be constructed from~$G$ and~$X$ as above. Let~$\delta(H,S,T)$ denote the minimum size of an~$(S,T)$ vertex cut in~$H$. The minimum size of~$Y\subseteq V$ such that~$G-Y$ is bipartite equals the minimum of~$|X \setminus U| + \delta(G'-\X{X \setminus U},S,T)$ over all subsets~$U$ of~$X$ and all valid splits~$(S,T)$ of~$U$.
\end{lemma}

Clearly, given this result, one  
can find an optimal bipartization set for~$G$ by looping over the~$3^{|X|}$ options for~$U$,~$S$, and~$T$.
In particular, one is not limited to using~$X$ of size~$k+1$ but, sacrificing runtime, a single run of the iterative compression routine on~$G$ and an approximate solution~$X$ for~$G$ suffices. In the next subsection, this setting will be used to give a polynomial kernelization of \OCT.

For proofs of Lemma~\ref{lemma:rsv}, see~\cite{ReedSV04} or one of the presentations subsequently given by other authors~\cite{Huffner09,LokshtanovSS09}.
For variation, we will now sketch an alternative approach to showing the result.

Instead of a graph, we view the OCT instance as a 2-SAT formula~$F$ containing only constraints~$(x=y)$ and~$(x\neq y)$.
Clearly, for any graph~$G$, if we replace every edge~$\{u,v\}$ in~$G$ by a constraint~$(u \neq v)$, then we get a formula~$F$
over~$V$ which is satisfiable if and only if~$G$ is bipartite, and this holds for every induced subgraph of~$G$ as well.
Thus, the problem reduces to deleting~$k$ variables~$Z$ of~$F$ and their incident constraints such that the remaining formula~$F-Z$ is satisfiable.

Now, observe that we can \emph{negate} a variable~$v$ in~$F$ by changing~$(u=v)$-constraints to~$(u \neq v)$-constraints 
and vice versa.  This does not affect the satisfiability of~$F-Y$ for any variable set~$Y$.
Thus, negate variables in~$F$ so that~$F-X$ is satisfied by the all-zero assignment. We now observe that the only remaining disequality
constraints~$(x \neq y)$ of~$F$ are incident to~$X$.  By deleting or assigning every~$x \in X$, we create smaller formulas~$F'$
containing only equality constraints~$(u=v)$ and assignments~$(u=0)$ or~$(u=1)$.
This trivially reduces to a vertex cut problem in a graph, which would conclude the proof of the FPT result.

To get from here to the construction of~$G'$ above, consider the effects of splitting variables~$x \in X$ in~$F$
into distinct variables~$x$ and~$\neg x$ representing its \emph{literals}, and replace constraints~$(x \neq y)$
by~$(\neg x = y)$. Clearly, in any assignment we must require~$(x \neq \neg x)$. 
Applying all this to the bipartition~$S_1 \cup S_2$ of~$G-X$ will show the equivalence of the result.

\subsection{Kernelization}\label{section:actualkernel}

Now, we give the kernelization.  
We begin by describing a compression procedure for OCT, by applying the matroid tools of Section~\ref{section:newmatroidsection}
to Lemma~\ref{lemma:rsv} above.  The result is a randomized polynomial-time compression procedure with one-sided 
error, consisting of the following steps.  Let an instance~$(G,k)$ of OCT and an error parameter~$\epsilon>0$ be given.
\begin{enumerate}
\item If~$k\leq \log n$, run the Reed-Smith-Vetta algorithm in time~$\BigOh(3^kkmn)=n^{\BigOh(1)}$ (polynomial in~$n$) and return a constant-size \yes- or \no-instance accordingly. \label{kernel:step:rsv}
\item Otherwise, if~$k>\log n$, let~$X$ be an approximate solution of ratio~$\BigOh(\sqrt{\log n})=\BigOh(k^{1/2})$, provided by an algorithm due to Agarwal et al.~\cite{AgarwalCMM05-approx}.
      Unless~$|X|=\BigOh(k^{3/2})$, answer \no as there cannot be a solution of size at most~$k$. (If~$|X|\leq k$ then answer \yes.)\label{kernel:step:approx}
\item Create the auxiliary graph~$G'$ from~$G$ and~$X$, as
in Section~\ref{section:rsvalg}.  
      Let~$X'=\{x_1,x_2\mid x\in X\}$.\label{kernel:step:auxiliarygraph}
\item Apply Corollary~\ref{corollary:smallrepresentedmatroid} to~$G'$ with terminal set~$X'$ and error parameter~$\epsilon/2$,
      creating a matrix~$A$. 
\item Output~$(A,k)$ as a polynomial-sized compression.\label{kernel:step:output}
\end{enumerate}
The total coding size of the matrix~$A$ is cubic in~$|X|$ (up to factors
logarithmic in~$1/\epsilon$). By the size guarantee in Step~\ref{kernel:step:approx}, we
get an~$\BigOh(k^{4.5})$-sized compression for OCT.
Note that~$(A,k)$ may be interpreted as an instance of an (artificial)
decision problem (implicitly defined in the proof of Lemma~\ref{lemma:compresscompress}).  Clearly, this problem is in \NP, allowing
us to reduce back to OCT to complete the kernelization (using the implicit Karp reduction as discussed in Section~\ref{section:prelim:kernels}).

For \EBip, we may replace Steps~\ref{kernel:step:rsv} and~\ref{kernel:step:approx} by the~$\BigOh(\log \textrm{OPT})$-approximation of 
Avidor and Langberg~\cite{AvidorL07-approx-klogk}, followed by an easy
reduction from \EBip to OCT, for a compression of size~$\tilde \BigOh(k^3)$.
It is an interesting question whether a polylog(OPT)-approximation is possible for OCT, 
as this would give us a~$\tilde \BigOh(k^3)$-sized compression for OCT.

Now we give our central compression result.

\begin{lemma} \label{lemma:compresscompress}
Let~$(G,k)$ be an instance of OCT,~$X$ a bipartization set for~$G$, and~$\epsilon>0$ be given.
Then there is a  
randomized compression of~$(G,k)$ to 
size~$\BigOh(|X|^2(|X|+\log 1/\epsilon))$
with one-sided error, producing no false negatives.
The error probability is bounded by~$\epsilon$,
and the running time is polynomial in~$|G|$ and~$\log 1/\epsilon$. 
\end{lemma}
\begin{proof}
The algorithm proceeds as Steps~\ref{kernel:step:auxiliarygraph}-\ref{kernel:step:output} of the kernelization algorithm,
creating first an auxiliary graph~$G'$ with a terminal set~$X'$ of size~$2|X|$, 
and then invoking Corollary~\ref{corollary:smallrepresentedmatroid} on~$G'$,~$X'$, and~$\epsilon/2$.
Let the resulting matrix be~$A$; our compression output is then~$(A,k)$.
The running time and output size are given by Corollary~\ref{corollary:smallrepresentedmatroid};
we only have to argue that~$(A,k)$ contains all the information needed
to decide the status of the OCT instance~$(G,k)$. 
By Step~\ref{kernel:step:approx}, we assume~$|X|>k$.

Recall the definition~$\X{U}=\{x_1,x_2: x \in U\}$ for~$U \subseteq X$.
By Lemma~\ref{lemma:rsv}, we need for all~$U\subseteq X$ the minimum~$(S,T)$ vertex cut size in~$G'-\X{X\setminus U}$, where~$S$ and~$T$ range over all valid splits of~$U$.
Clearly, if the minimum~$(S,T)$ vertex cut size in~$G'-R$ is~$\lambda$
then there is a set~$T'\subseteq T$ of size~$\lambda$ such that~$T'$ is linked to~$S$ in~$G'-R$, 
i.e., such that there are~$\lambda=|T'|$ vertex disjoint paths from~$S$ to~$T'$ in~$G'-R$ (using cut-flow duality).
This can be obtained from the matrix~$A$ by testing independence of all sets~$I$ which correspond 
(as in Corollary~\ref{corollary:smallrepresentedmatroid}) to choices~$S,T',R$ with~$T'\subseteq T$ 
and~$R=\X{X \setminus U}$. Note that whether an~$(S,T)$-cut may delete~$S$ and~$T$ makes no difference for the algorithm using Lemma~\ref{lemma:rsv}.

The behavior and one-sidedness of the error also follows from Corollary~\ref{corollary:smallrepresentedmatroid}.
\end{proof}

With the previously described algorithm and the approximation results
we get the following.

\begin{theorem}
There is a randomized~$\BigOh(k^{4.5})$-compression for \OCT and a 
randomized~$\tilde \BigOh(k^3)$-compression for \EBip, with one-sided error with no
false negatives and failure probability exponentially small in~$k$. 
\end{theorem}

The target problems of the compressions are constrained minimization problems over the rank of a matroid. We omit the precise definitions, but it should be clear that the problems can be made well-defined, and that they are in \NP.
As discussed in Section~\ref{section:prelim:kernels}, using \NP-completeness of \OCT and \EBip, we get the following polynomial \coRP-kernelization results.

\begin{corollary}
\OCT, \EBip, \textsc{Balanced Subgraph}, and \textsc{Tanglegram Layout} (see~\cite{BockerHTW09}) have polynomial \coRP-kernelizations.
\end{corollary}

\subsection{A Note on Lower Bounds} \label{section:kernellower}

Finally, we remark that although kernelizations and kernelization lower bounds
are usually expressed in terms of deterministic results, the type of
randomized kernelizations we produce here (i.e., one-sided error with no false
negatives) does in fact fit within the lower bounds framework of Bodlaender et
al.~\cite{BodlaenderDFH09} and Fortnow and Santhanam~\cite{FortnowS11}, as the proofs implicitly exclude
also co-nondeterministic kernels.
Since our \coRP kernelization is a special case of this,
our tools cannot be used to escape the lower bounds.
Dell and van Melkebeek~\cite{DellM10} noted the connection to co-nondeterministic kernels,
and brought it further by giving
lower bounds in terms of the amount of communication needed to solve
a problem in an oracle setting, where a polynomially bounded but
co-nondeterministic player communicates with an computationally unbounded oracle.\footnote{Harnik and Naor~\cite{HarnikN10} credit an extension of~\cite{FortnowS11} to any one-sided error to unpublished work of Chen and M\"uller.}
They provide concrete lower bounds for various problems, among others implying the following.
\begin{theorem}[\cite{DellM10}]
Let~$\epsilon>0$. No co-nondeterministic kernel or compression for OCT can achieve a total size of~$\BigOh(k^{2-\epsilon})$ unless \containment and the polynomial hierarchy collapses.
\end{theorem}

It seems difficult to go below an~$\BigOh(k^3)$ bound using our methods.  Thus, we
leave it as an open question whether the upper or the lower bound on the total
compression size can be improved for \OCT and \EBip.

\section{Conclusion}\label{section:conclusion}

We have presented randomized polynomial kernelizations for \OCT and \EBip. The
key contribution is the introduction of matroids into kernelization, by
encoding the compression step of 
the Reed-Smith-Vetta
algorithm for \OCT~\cite{ReedSV04}, by means of a matroid.  
This leads to a compression of the problem into size~$\BigOh(k^{4.5})$ for \OCT
and~$\tilde \BigOh(k^3)$ for \EBip, which is easily turned into a kernelization 
by back-reductions to the original problems. The kernelization
has one-sided error, producing no false negatives, and the failure rate can be
made exponentially small in~$k$ at only a constant factor cost to the size.
While this essentially settles the question about existence of polynomial kernels,
the more practical result seems to be the output of the compression. 
Not only is compression to any set the
more robust notion (cf.~\cite{DellM10}); the target problem is native in one 
of the most well-studied areas of mathematics and computer science.
The compression may also point the way for where to look for direct, combinatorial kernelizations for the problem.

It is interesting that the results of Fortnow and Santhanam~\cite{FortnowS11}
can be seen to exclude also co-nondeterministic
compressions~\cite{DellM10}. Thus, our technique is not a way of avoiding the
lower bounds given by~\cite{BodlaenderDFH09,FortnowS11}, but a way of settling
problems for which neither such lower bounds nor a polynomial kernelization or
compression are known.

We close with some open problems. It is still interesting whether there exist
deterministic polynomial kernelizations for \OCT and \EBip, either as a
derandomization of our methods, or (which would have independent interest) as
a properly combinatorial kernelization.  Additionally, for both problems, there is the question
of the correct size bound.  We note that
a~$\BigOh$(polylog(OPT))-approximation for OCT, which is consistent with
approximation theory lower bounds, would improve our result to~$\tilde
\BigOh(k^3)$, but this still leaves a gap to the~$\BigOh(k^{2-\epsilon})$ lower
bound given by Dell and van Melkebeek~\cite{DellM10}.
Finally, existence of (randomized) polynomial kernels is an
exciting question for several related problems, including 
\DFVS, \textsc{Multiway Cut}, and \textsc{Edge} and \textsc{Vertex Multicut}.
The robust way in which the introduction of matroids into kernelization
helped to settle the question for \OCT gives reason to believe that it will play
a key role for some of these problems as well.

\bibliographystyle{abbrv}

\end{document}